\documentclass[final]{siamltex}
\usepackage{graphicx,epsfig}
\title{A free-boundary problem for concrete carbonation: Rigorous justification of the $\sqrt{{t}}$-law of propagation}

\author{Toyohiko Aiki\thanks{Department of Mathematics, Faculty of Education, Gifu University,
Yanagido 1-1, Gifu, 501-1193,    Japan ({\tt aiki@gifu-u.ac.jp}).}
        \and Adrian Muntean\thanks{CASA - Centre for Analysis, Scientific computing and Applications,
Department of Mathematics and Computer Science, Institute for
Complex Molecular Systems (ICMS), Eindhoven University of
Technology, PO Box 513, 5600 MB Eindhoven, The Netherlands ({\tt
a.muntean@tue.nl}).}}

\begin{document}

\maketitle

\begin{abstract}
We study a one-dimensional free-boundary problem describing the
penetration of carbonation fronts (free reaction-triggered
interfaces) in concrete. A couple of decades ago, it was observed
experimentally that the penetration depth versus time curve (say
$s(t)$ vs. $t$) behaves like $s(t)=C\sqrt{t}$  for sufficiently
large times  $t > 0$ (with $C$ a positive constant). Consequently,
many fitting arguments solely based on this experimental law were
used to predict the large-time behavior of  carbonation fronts in
real structures, a theoretical justification of the $\sqrt{t}$-law
being lacking until now.
%This is the place where our paper contributes:

The aim of this paper is to fill this gap by justifying rigorously
the experimentally guessed asymptotic behavior. We have previously
proven the upper bound $s(t)\leq C'\sqrt{t}$ for some constant $C'$;
now we show  the optimality of the rate by proving the right
nontrivial lower estimate, i.e. there exists $C''>0$ such that
$s(t)\geq C''\sqrt{t}$. Additionally, we obtain weak solutions to
the free-boundary problem for the case when the measure of the
initial domain vanishes. In this way, our mathematical model is now
allowing for the appearance of a moving carbonation front -- a
scenario that until was hard to handle from the analysis point of
view.
\end{abstract}

\begin{keywords}
Free-boundary problem, concrete carbonation, large-time behavior,
$\sqrt{t}$-law of propagation, appearance of a carbonation front,
phase change;
\end{keywords}

\begin{AMS}
35R35, 35B40, 80A22
\end{AMS}

\pagestyle{myheadings} \thispagestyle{plain} \markboth{T. AIKI AND
A. MUNTEAN}{Justification of the $\sqrt{t}$-law}

\section{Introduction}\label{intro}

\subsection{Background}
Environmental impact on concrete parts of buildings results in a
variety of unwanted chemical and chemically-induced mechanical
changes. The bulk of these changes leads to damaging and
destabilization of the concrete itself or of the reinforcement
embedded in the concrete. One important destabilization factor is
the drop in pH near the steel bars induced by carbonation of the
alkaline constituents; see for instance
\cite{Ishida1,Ishida2,Thiery} and \cite{Kris} for technical details
and \cite{Mun1,Ai-Mun5} for an introduction to the mathematical
modeling of the situation\footnote{Remotely related mathematical
approaches of similar reaction-diffusion scenarios have been
reported, for instance, in \cite{Hilhorst, Natalini,Du}.}. The
destabilization is caused by atmospheric carbon dioxide diffusing in
the dry parts and reacting in the wet parts of the concrete pores.
The phenomenon is considered as one of the major processes inducing
corrosion in concrete. A particular feature of carbonation is the
formation of macroscopic sharp reaction interfaces or thin reaction
layers that progress into the unsaturated concrete-based materials.
The deeper cause for the formation of these patterns is not quite
clear, although the major chemical and physical reasons seem to be
known.

Mathematically, the proposed model is a coupled system of
semi-linear partial differential equations posed in a single 1D
moving domains. The moving interface  (front position in 1D) is
assumed to be triggered by a fast chemical reaction  -- the
carbonation reaction. Non-linear transmission conditions of
Rankine-Hugoniot type are imposed across the inner boundary that
separates the carbonated regions from the uncarbonated ones. The
movement of the carbonated region is determined via a non-local
dynamics law.

The key objective  is not only to understand the movement of a
macroscopic sharp reaction front in concrete but rather to {\em
predict the penetration depth after a sufficient large time}.

 A couple of decades ago, it was
observed experimentally that the penetration depth versus time curve
(say $s(t)$ vs. $t$) behaves like $s(t)=C\sqrt{t}$  for sufficiently
large times  $t > 0$ (with $C$ a positive constant). Consequently,
many fitting arguments solely based on this experimental law were
used to predict the large-time behavior of carbonation fronts in
real structures, a theoretical justification of the $\sqrt{t}$-law
being lacking until now.

This is the place where our paper contributes: We want to fill this
gap by justifying rigorously the experimentally guessed asymptotic
behavior.

\subsection{Basic carbonation scenario--a moving one-phase approach}
 We study a one-dimensional free boundary problem
system arising in the modeling of concrete carbonation problem. We
consider that the concrete occupies the infinite interval
$(0,\infty)$ and that there exists a sharp interface $x = s(t)$, $t
> 0$ separating the carbonated from the uncarbonated zone. The whole process can be seen as a solid-solid phase change; see the two colors in Fig. \ref{fig1} (left).
One color points out to $CaCO_3$ (carbonated phase), while the other
one indicates $Ca(OH)_2$ (uncarbonated phase). The zone of interest
is only one of the solid phases, namely the carbonated zone. We
denote it by $Q_s(T)$ and, in mathematical terms, this is defined by
$Q_s(T) := \{ (t,x): 0 < t < T, 0 < x < s(t)\}$ for some $T
> 0$. Throughout this paper $u$ and $v$ denote the mass
concentrations of CO$_2$ in air and water, respectively. As
mentioned in \cite{Ai-Mun5}, $s$, $u$ and $v$ satisfy the following
system P $= $P$(s_0, u_0, v_0, g, h)$ (\ref{eq:a10}) $\sim$
(\ref{eq:a20}):
\begin{eqnarray}
& &
u_t - (\kappa_1 u_x)_x = f(u,v) \quad \mbox{ in } Q_s(T), \label{eq:a10} \\
& & v_t - (\kappa_2 v_x)_x = - f(u,v) \quad  \mbox{ in } Q_s(T), \\
& & u(t,0) = g(t), v(t,0) = h(t) \quad \mbox{ for } 0 \leq t \leq T, \\
& & s'(t) =  \psi(u(t,s(t)))  \quad \mbox{ for } 0 < t < T, \label{FBC} \\
& &  - \kappa_1 u_x(t,s(t)) = \psi(u(t,s(t)))  + s'(t) u(t,s(t))   \quad \mbox{ for } 0 < t < T, \\
& &  - \kappa_2 v_x(t,s(t)) = s'(t) v(t, s(t)) \quad \mbox{ for } 0 < t < T, \\
& & s(0) = s_0 \mbox{ and } u(0,x) = u_0, v(0,x) = v_0 \quad \mbox{
for } 0 < x < s_0, \label{eq:a20}
\end{eqnarray}
where $\kappa_1$ (resp. $\kappa_2$) is a diffusion constant of
CO$_2$ in air (resp. water), $f(u,v) := \beta(\gamma v - u)$ is an
effective Henry's law, where $\beta$ and $\gamma$ are positive
constants, $g$ and $h$ are given functions corresponding to boundary
conditions for $u$ and $v$, respectively, $\psi(r) := \alpha
|[r]^+|^p$ for $r \in R$ describes the rate of the carbonation
reaction, where $p \geq 1$ and $\alpha$ is a positive
constant\footnote{The exponent $p$ is sometimes called order of the
chemical reaction, while  the parameter $\alpha$ is just a
proportionality constant. Its sensitivity with respect to the model
output $(s,u,v)$ has been studied numerically in \cite{Chem}.}.
$s_0\geq 0$ is the initial position of the free boundary, while
$u_0$ and $v_0$ are the initial concentrations.

First mathematical models with free boundaries for describing the
concrete carbonation process have been proposed by Muntean and
B\"ohm in \cite{Mun1,Mun2}, where the first mathematical results
concerning the global existence and uniqueness of weak solutions as
well as the stability of the solutions with respect to data and
parameters have been investigated.  Recently, we have improved their
results by focussing a reduced free-boundary model still able to
capturing the basic features of the carbonation process; see
\cite{Ai-Mun5} for the reduced model and  \cite{Ai-Mun1,Ai-Mun4} for
the list of the new theorems on the existence and uniqueness of weak
solutions to P. This model is in some sense minimal: It includes the
transport of species (diffusion), their averaged transfer across
air-water interfaces (the Henry law), as well as fast reaction (with
an indefinitely large chemical compound - "the concrete"). We have
used further the advantageous structure of the reduced model to
study the large-time behavior of the penetration depths. Basically,
we started to wonder whether the experimentally known {\em
$\sqrt{t}$-law}
$$s(t) = \bar{C}\sqrt{t} \mbox{ for } t >0,$$ where $\bar{C}$ is a
positive constant, is true or not \cite{limitations}. Let us comment
a bit on the context: It was shown in \cite{Thiery} (pp. 193-199)
that the carbonation front behaves like a similarity solution to a
one-phase Stefan-like problem \cite{Rubinstein}. Using
matched-asymptotics techniques, the fast-reaction limit (for large
Thiele moduli) done in \cite{Mec} for a reaction-diffusion system
also led to a $\sqrt{t}$-behavior of the carbonation front
supporting experimental results from \cite{Ishida1,Ishida2}, e.g.
\begin{figure}[hptb]
\begin{minipage}{\linewidth}
\centerline{
\begin{minipage}{.30\linewidth}
\begin{center}\includegraphics[width=\linewidth]{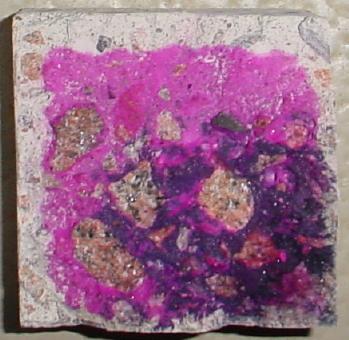}\\\end{center}
\end{minipage}
\begin{minipage}{.51\linewidth}
\begin{center}\includegraphics[width=1.2\linewidth]{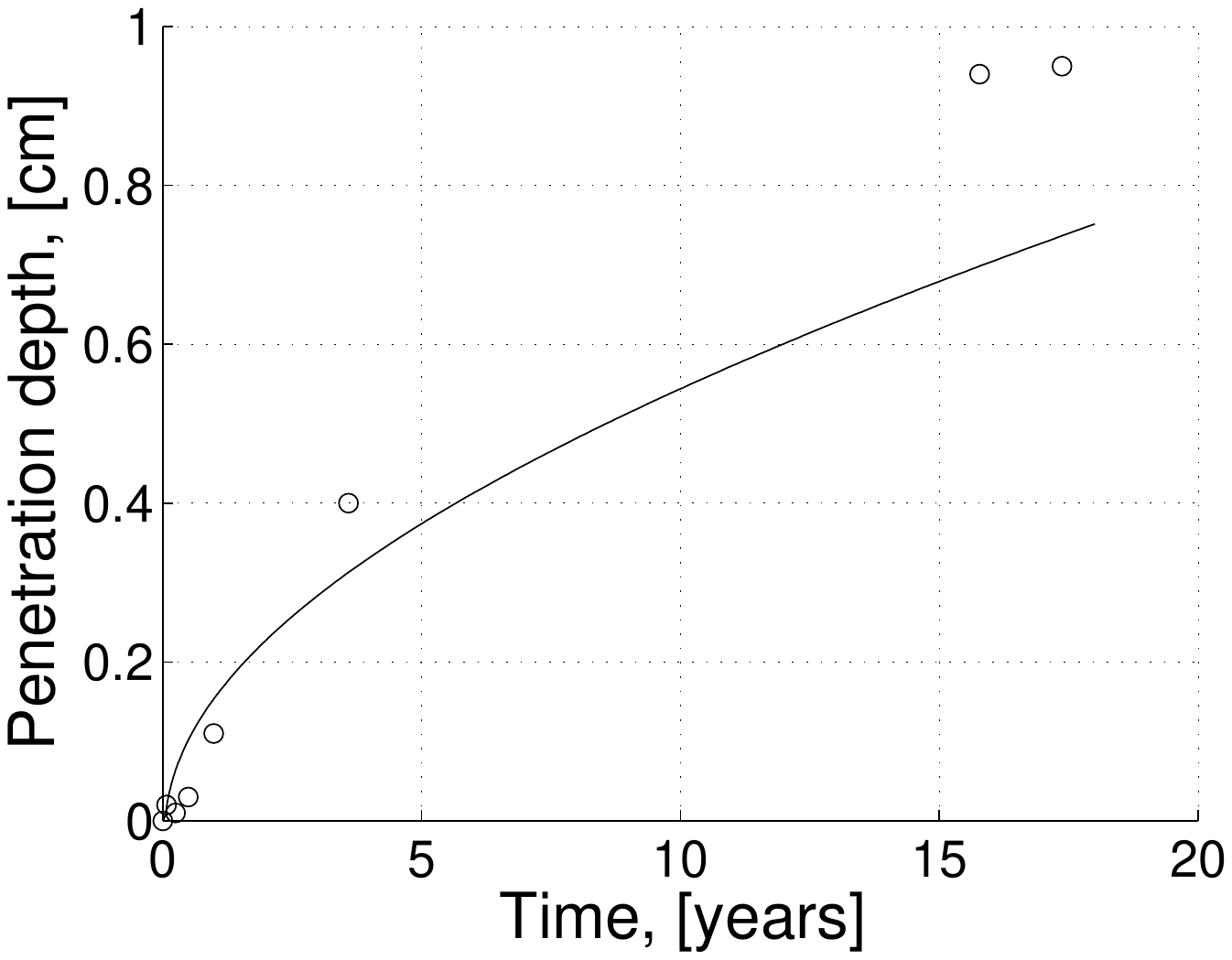}\\\end{center}
\end{minipage}
}
\end{minipage}\caption{(Left) Typical result of the phenolphthalein test
  on a partially carbonated sample (Courtesy of Prof. Dr. Max Setzer, University
  of Duisburg-Essen, Germany). The dark region indicates the
  uncarbonated part, while the brighter one points out the carbonated part. The
  two regions are separated by a sharp interface moving inwards the material. In this colorimetric
  test, this macroscopic
  interface corresponds to a drop in pH below 10. (Right) Computed interface positions vs. measured penetration depths \cite{Chem}.}\label{fig1}
\end{figure}
On the other hand, experimental results from \cite{Kris} indicate
that, depending on the type of the cement, a variety of $t^\beta$
front behaviors with $\beta \neq \frac{1}{2}$ are possible.
Furthermore, Souplet, Fila and collaborators (compare
\cite{Souplet,Fila}) have shown that, under certain conditions,
non-homogeneous Stefan-like problem can lead to asymptotics like
$s(t)\sim t^\frac{1}{3}$. Somehow, the major question remains:

{\em What is  the correct asymptotics of the carbonation front
propagation?}

The main result of our preliminary investigations (based on the
reduced FBP) is reported in \cite{Ai-Mun6} and supports the fact
that
$$ s(t) \to \infty \mbox{ as } t \to \infty \mbox{ and }
 s(t) \leq C' \sqrt{t} \mbox{ for } t \geq 0,  $$
where $C'$ is a positive constant. Moreover, in this paper we
establish a result on the lower estimate for the free boundary $s$
as follows: For some positive constant $c$
\begin{equation}
s(t) \geq c \sqrt{t} \mbox{ for } t \geq 0. \label{result1}
\end{equation}
This estimate combined with the corresponding lower one would
immediately guarantee the correctness of the $\sqrt{t}$-law from a
mathematical modeling point of view. In section \ref{large} we
derive the missing lower bound.

Note that since, generally, $k_1\gg k_2$ and $\gamma v\neq u$, the
system (\ref{eq:a10})--(\ref{eq:a20}) cannot be reduced to a scalar
equation, where the use of Green functions representation
\cite{Cannon,Rubinstein} would very much facilitate the obtaining of
non-trivial lower bounds on concentrations, and hence, on the free
boundary velocity. Furthermore, by using the similar method as the
one used in the proof of (\ref{result1}), we can construct a weak
solution to P satisfying  $s_0 = 0$. It is worth mentioning that
Fasano and Primicerio (cf. \cite{FP2}, e.g.) have investigated a
one-phase Stefan problem when the measure of the initial domain
vanishes.  In their proof the comparison principle is used in an
essential manner. However, for our problem P we do not have any
comparison theorem for the free boundary. Our idea here is to
develop a method to obtain improved uniform estimates for solutions
and then use these estimates to prove the existence of weak
solutions for the case $s_0 = 0$. This program is realized in
section \ref{front}. There are neither physical nor mathematical  reasons to believe that uniqueness of weak solutions  for the case $s_0 = 0$ would not hold. However, since our fixing-domain technique is not applicable anymore, the uniqueness seems to be difficult to prove.

\section{Large-time behavior of the free boundary}
\label{large} In order to give a statement of our result on the
large-time behavior of (weak) solutions, we consider the problem P
posed in the cylindrical domain $Q(T) := (0,T) \times (0,1)$. To
this end, we use the following change of variables:

Let
\begin{equation} \label{transf}
\bar{u}(t,y) = u(t,s(t)y) \mbox{ and } \bar{v}(t,y) = v(t,s(t)y)
\mbox{  for } (t,y) \in Q(T). \label{change}
\end{equation}
 Then, it holds that
  \begin{eqnarray*}
& & {\bar u}_t - \frac{\kappa_1}{s^2} \bar{u}_{yy} - \frac{s'}{s} y
\bar{u}_y = f(\bar{u}, \bar{v})
 \quad \mbox{ in } Q(T),  \\
& & {\bar v}_t - \frac{\kappa_2}{s^2} \bar{v}_{yy} - \frac{s'}{s} y
\bar{v}_y = - f(\bar{u}, \bar{v})
 \quad \mbox{ in } Q(T), \\
& & \bar{u}(0,t) = g(t), \bar{v}(0,t) = h(t) \quad \mbox{ for } 0 < t < T, \\
& & s'(t) =  \psi(\bar{u}(t,1))  \quad \mbox{ for } 0 < t < T, \\
& & - \frac{\kappa_1}{s(t)} {\bar u}_y(t,1) = s'(t)  \bar{u}(t,1) +
  s'(t)  \quad \mbox{ for } 0 < t < T, \\
& & - \frac{\kappa_2}{s(t)} \bar{v}_y(t,1) = s'(t) \bar{v}(t, 1)
         \quad \mbox{ for } 0 < t < T, \\
& & s(0) = s_0, \bar{u}(0,y) = \bar{u}_0(y), \bar{v}(0,y) =
\bar{v}_0(y) \quad \mbox{ for } 0 < y < 1,
\end{eqnarray*}
where $\bar{u}_0(y) = u_0(s_0 y)$ and $\bar{v}_0(y) = v_0(s_0 y)$
for $y \in [0,1]$.

For simplicity, we introduce some notations as follows: $H  :=
L^2(0,1)$, $X := \{z \in H^1(0,1): z(0) = 0\}$,  $X^*$ is the dual
space of $X$, $$V(T) := L^{\infty}(0,T; H) \cap L^2(0,T; H^1(0,1))$$
and $$V_0(T) := V(T) \cap L^2(0,T; X),$$ and $(\cdot,\cdot)_H$ and
$\langle \cdot, \cdot \rangle_X$ denote the usual inner product of
$H$ and the duality pairing between $X$ and $X^*$, respectively.

First of all, we define a weak solution of P$(s_0, u_0, v_0, g, h)$.
To do this, we use a similar concept of weak solution as the one
introduced in \cite{Ai-Mun1}.

\begin{definition}
Let $s$ be a function on $[0,T]$ and $u$, $v$ be functions on
$Q_s(T)$ for $0 < T < \infty$, and $\bar{u}$ and $\bar{v}$ be
functions defined by (\ref{change}). We call that a triplet $\{s, u,
v\}$ is a weak solution of P on $[0,T]$
if the conditions (S1) $\sim$ (S5) hold: \\
(S1) $s \in W^{1,\infty}(0,T)$ with $s > 0$ on $[0,T]$,
$(\bar{u},\bar{v}) \in (W^{1,2}(0,T; X^*) \cap V(T) \cap
L^{\infty}(Q(T)))^2$.
\\
(S2) $\bar{u} - g, \bar{v} - h \in L^2(0,T; X)$, $u(0) = u_0$ and
$v(0) =  v_0$.
\\
(S3) $s'(t) = \psi(u(t,s(t))$ for a.e. $t \in [0,T]$ and $s(0) =
s_0$.
$$
\int_0^{T} \langle \bar{u}_t, z \rangle_X dt + \int_{Q(T)}
\frac{\kappa_1}{s^2} \bar{u}_y z_y dy dt + \int_0^{T}
 \frac{s'}{s} (\bar{u}(\cdot,1) + 1 )z(\cdot,1) dt  \leqno{\mbox{(S4)}}
 $$
$$ = \int_{Q(T)} (f(\bar{u},\bar{v}) + \frac{s'}{s} y \bar{u}_y) z dydt
 \quad \mbox{ for }
 z \in V_0(T).
$$
$$
\int_0^{T} \langle \bar{v}_t, z \rangle_X dt + \int_{Q(T)}
\frac{\kappa_2}{s^2} \bar{v}_y z_y dy dt + \int_0^{T}
 \frac{s'}{s} \bar{v}(\cdot,1)  z(\cdot,1) dt  \leqno{\mbox{(S5)}}
 $$
$$ = \int_{Q(T)} (- f(\bar{u},\bar{v}) + \frac{s'}{s} y \bar{v}_y) z dy dt
 \quad \mbox{ for }
 z \in V_0(T).
$$
Moreover, let $s$ be a function on $[0,\infty)$, and $u$ and $v$ be
functions on $Q_s := \{ (t,x)| t > 0, 0 < x < s(t)\}$. We say that
$\{s, u, v\}$ is a weak solution of P on $[0,\infty)$ if for any $T
> 0$ the triplet $\{s, u, v\}$  is a weak solution of P on $[0,T]$.
\end{definition}

Before recalling our results concerning the global existence and
uniqueness of weak solutions to P on the time interval $[0,T]$, $T >
0$, we give the following assumptions for the involved data and
model parameters:

(A1) $f(u,v) =  \beta( \gamma v - u)$ for any $(u,v) \in {\mathbf
R}^2$ where $\beta$ and $\gamma$ are positive constants.

(A2) $g, h \in W^{1,2}_{loc}([0,\infty)) \cap L^{\infty}(0,\infty)$,
and $g \geq 0$ and $h \geq 0$ on $(0,\infty)$.

(A3) $u_0 \in L^{\infty}(0,s_0)$ and $v_0 \in L^{\infty}(0,s_0)$
with $u_0 \geq 0$ and $v_0 \geq 0$ on $(0,s_0)$.

\begin{theorem} \label{previous}
(cf. \cite[Theorems 1.1 and 1.2, Lemma 4.1]{Ai-Mun1}) If   (A1)
$\sim$ (A3) hold, then P has one and only one weak nonnegative
solution on $[0,\infty)$.
\end{theorem}

The next theorem is the main result of this paper.

\begin{theorem} \label{main1}
If $g(t) = g_*$, $h(t) = h_*$ for $t \in [0,\infty)$, where $g_*$
and $h_*$ are positive constants with $\gamma h_* = g_*$, and (A1)
and (A3) hold, then there exists a positive constant $c$ such that
$$ s(t) \geq c \sqrt{t} \quad \mbox{ for } t \geq 0. $$
\end{theorem}

\vskip 12pt The proof of Theorem \ref{main1} relies on three
technical lemmas. We give these auxiliary results in the following.

\begin{lemma}
\label{lem1} (cf. \cite[Lemma 3.3]{Ai-Mun6}) If  (A1) $\sim$ (A3)
hold, then a weak  solution $\{s, u, v\}$ on $[0,\infty)$ satisfies
\begin{eqnarray}
& & \int_0^{s(t)} x u(t) dx + \frac{1}{2} |s(t)|^2 + \kappa_1
\int_0^t u(\tau, s(\tau)) d\tau
   \nonumber \\
& & + \int_0^{s(t)} x v(t) dx + \kappa_2 \int_0^t \int_0^{s(\tau)}
v_x(\tau,x) dx d\tau
  \nonumber \\
& = & \int_0^{s_0} x u_0 dx + \frac{1}{2} |s_0|^2 + \int_0^{s_0} x
v_0 dx +
 \kappa_1 \int_0^t g(\tau) d\tau \quad \mbox{ for } t \geq 0. \nonumber
\end{eqnarray}
\end{lemma}

\begin{proof}
Let $T > 0$. In (S4) we can take $z(t) =  s^2(t)y$ for $t \in [0,T]$
so that  we have
\begin{eqnarray}
& & \int_0^{T} \langle \bar{u}_t(t), s^2(t)y \rangle_X dt +
\int_{Q(T)} \kappa_1 \bar{u}_y(t)  dy dt
%% \nonumber \\
 + \int_0^{T}
 s'(t)s(t) (\bar{u}(t,1) + 1 )  dt
 \nonumber \\
 &= &
\int_{Q(T)} (y f(\bar{u}(t),\bar{v}(t)) + \frac{s'(t)}{s(t)} y^2
\bar{u}_y(t)) s^2(t)  dydt. \label{b0}
\end{eqnarray}
Here, we note that
\begin{eqnarray}
& &  \int_0^{T} \langle \bar{u}_t(t), s^2(t)y \rangle_X dt
\nonumber \\
& = & - 2 \int_0^T \int_0^1 \bar{u}(t)s'(t) s(t) y dy dt
  + \int_0^1 \bar{u}(T) s^2(T)y dy
  - \int_0^1 \bar{u}(0) s^2(0)y dy
 \nonumber \\
& = & - 2 \int_0^T \int_0^{s(t)}  u(t)  \frac{s'(t)}{s(t)} x dx  dt
  + \int_0^{s(T)}  u(T) x dx
  - \int_0^{s_0} u_0 x dx; \label{b2}
\end{eqnarray}
\begin{eqnarray}
\int_{Q(T)} \kappa_1 \bar{u}_y(t)  dy dt = \kappa_1 \int_0^T
(\bar{u}(t,1)  - g(t)) dt; \label{b3}
\end{eqnarray}
\begin{eqnarray}
\int_0^{T}
 \frac{s'(t)}{s(t)} (\bar{u}(t,1) + 1 ) s^2(t) dt
= \int_0^{T}  s'(t)s(t)\bar{u}(t,1) dt + \frac{1}{2}(s^2(T) -
s_0^2);
\end{eqnarray}
and
\begin{eqnarray}
& & \int_{Q(T)}  \frac{s'(t)}{s(t)} y^2 \bar{u}_y(t) s^2(t)  dydt  \nonumber \\
& = & -2 \int_{Q(T)} s'(t)s(t)  y \bar{u}(t)   dydt
   + \int_0^T s'(t) s(t) \bar{u}(t,1) dt \nonumber \\
& = & -2 \int_{Q(T)} \frac{s'(t)}{s(t)}  u(t)   dxdt
   + \int_0^T s'(t) s(t) x \bar{u}(t,1) dt. \label{b4}
\end{eqnarray}
By substituting (\ref{b2}) $\sim$ (\ref{b4}) into (\ref{b0}) we see
that
\begin{eqnarray*}
 & &   \int_0^{s(T)}  x u(T)  dx + \frac{1}{2}s^2(T) + \kappa_1 \int_0^T u(t,s(t))  dt
\\
&  = &
  \int_0^{s_0}  x u_0  dx + \frac{1}{2}s_0^2 + \kappa_1 \int_0^T g(t) dt
 + \int_{Q_s(T)} f(u,v) x dxdt.
\end{eqnarray*}
Similarly, it follows from (S5) with $z(t) =  s^2(t)y$ that
\begin{eqnarray*}
 \int_0^{s(T)}  x v(T)  dx  + \kappa_2 \int_{Q_s(T)} v_x(\tau,x) dx dt =
  \int_0^{s_0}  x v_0  dx - \int_{Q_s(T)} f(u,v) x dxdt.
\end{eqnarray*}
Adding these two equations leads to the end of the proof of this
lemma.
\end{proof}

Before starting off to providing a proof for the main result of the
paper (Theorem \ref{main1}), we wish to point out in Lemma
\ref{lem2} and in Lemma \ref{lem3} below that our free-boundary
problem allows for positive and uniformly bounded concentrations,
and also, that an energy-like inequality holds.

\begin{lemma}
\label{lem2} (cf. \cite[Lemma 3.2]{Ai-Mun6}) Assume (A1) $\sim$ (A3)
hold,  take positive numbers $g^*$ and $h^*$ satisfying $u_0 \leq
g^*$, $v_0 \leq h^*$ on $[0,s_0]$, $g \leq g^*$, $h \leq h^*$ on
$[0,\infty)$ and $g^* = \gamma h^*$, and let $\{s,u,v\}$ be a weak
solution of P on $[0,\infty)$. Then it holds that
$$ 0 \leq u \leq g^*, 0 \leq v  \leq h^* \mbox{ on } Q_s. $$
\end{lemma}

\begin{lemma}
\label{lem3} (cf. \cite[Lemma 3.4]{Ai-Mun6}) Under the same
assumptions as in Theorem \ref{main1} a weak solution $\{s,u,v\}$ of
P on $[0,\infty)$ satisfies
\begin{eqnarray*}
& & \frac{1}{2} \int_0^{s(t)} |u(t) - g_*|^2 dx + \frac{\gamma}{2}
\int_0^{s(t)} |v(t) - h_*|^2 dx
+  \frac{1}{2} \int_0^t |s'(\tau)|^{1+2/p} d\tau \\
& & + \kappa_1 \int_0^t \int_0^{s(\tau)} |u_x(\tau)|^2 dx d\tau +
 \gamma \kappa_2 \int_0^t \int_0^{s(\tau)} |v_x(\tau)|^2 dx d\tau  \\
& \leq &
 \frac{1}{2} \int_0^{s_0} |u_0 - g_*|^2 dx + \frac{\gamma}{2} \int_0^{s_0} |v_0 - h_*|^2 dx
\\
& & + \int_0^t s'(\tau) (\frac{1}{2} |g_*|^2 + g_* +
\frac{\gamma}{2} |h_*|^2) d\tau
 \quad \mbox{ for } t \geq 0.
\end{eqnarray*}
\end{lemma}

\subsection{Proof of Theorem \ref{main1}.} In this section, we give
 the proof of our main result.

\begin{proof}
Let $\{s, u, v\}$ be a weak solution of P on $[0,\infty)$, and $g^*$
and $h^*$ be positive constants defined in Lemma \ref{lem2}. First,
Lemma \ref{lem3} implies that
\begin{eqnarray}
& &  \gamma \kappa_2 \int_0^t \int_0^{s(\tau)} |v_x(\tau)|^2 dx d\tau \nonumber \\
& \leq &
 \frac{1}{2} \int_0^{s_0} (|u_0 - g_*|^2  + \gamma |v_0 - h_*|^2)  dx
 + \int_0^t s'(\tau) (\frac{1}{2} |g_*|^2 + g_* + \frac{\gamma}{2} |h_*|^2) d\tau
\nonumber \\
& \leq &
 \frac{1}{2} \int_0^{s_0} (|u_0 - g_*|^2  + \gamma |v_0 - h_*|^2)  dx \nonumber \\
& &
 +  \left(\frac{1}{2} |g_*|^2 + g_* + \frac{\gamma}{2} |h_*|^2\right) (s(t) - s_0) \quad \mbox{ for } t \geq 0.\nonumber
\end{eqnarray}
Hence, there is a positive constant depending on $u_0, v_0, g_*,
h_*$  and $s_0$ such that
\begin{eqnarray*}
\int_0^t \int_0^{s(\tau)} |v_x(\tau)|^2 dx d\tau
 \leq  C_1 + C_1 s(t)   \quad \mbox{ for } t \geq 0.
\end{eqnarray*}

Next,  on account of Lemma \ref{lem1} we see that
\begin{eqnarray}
& & \int_0^{s(t)} x u(t) dx + \frac{1}{2} |s(t)|^2 + \kappa_1
\int_0^t u(\tau, s(\tau)) d\tau
   \nonumber \\
& & + \int_0^{s(t)} x v(t) dx + \kappa_2 \int_0^t \int_0^{s(\tau)}
v_x(\tau,x) dx d\tau
  \nonumber \\
& = & \int_0^{s_0} x u_0 dx + \frac{1}{2} |s_0|^2 + \int_0^{s_0} x
v_0 dx +
 \kappa_1 \int_0^t g(\tau) d\tau  \nonumber \\
& \geq & \kappa_1 g_* t \quad \mbox{ for } t \geq 0. \nonumber
\end{eqnarray}
Here, we note that
$$ u(t,s(t)) = (\frac{s'(t)}{\alpha})^{1/p} \quad \mbox{ for } t \geq 0. $$
Then, by putting $M = \max\{g^*, h^*\}$ we obtain
\begin{eqnarray*}
 \kappa_1 g_* t
& \leq & 2M \int_0^{s(t)} x dx + \frac{1}{2} |s(t)|^2
+ \frac{\kappa_1}{\alpha^{1/p}} \int_0^t (s'(\tau))^{1/p} d\tau \\
& &  + \kappa_2 (\int_{Q_s(t)} |v_x|^2 dxd\tau )^{1/2}
(\int_{Q_s(t)}  dxd\tau )^{1/2} \quad \mbox{ for } t \geq 0.
\end{eqnarray*}
It is clear that
\begin{eqnarray}
 \frac{\kappa_1}{\alpha^{1/p}} \int_0^t (s'(\tau))^{1/p} d\tau
& \leq & \frac{\kappa_1}{\alpha^{1/p}} (\int_0^t s'(\tau) d\tau)^{1/p} t^{1-1/p} \nonumber \\
& \leq &  \frac{\kappa_1}{\alpha^{1/p}} s(t)^{1/p} t^{1-1/p} \nonumber \\
& \leq &  \frac{1}{4} \kappa_1 g_* t + C_2 s(t) \quad \mbox{ for } t
\geq 0, \label{i2}
\end{eqnarray}
where $C_2$ is some positive constant, and
\begin{eqnarray*}
& & \kappa_2 (\int_{Q_s(t)} |v_x|^2 dxd\tau )^{1/2}
(\int_{Q_s(t)}  dxd\tau )^{1/2} \nonumber \\
& \leq &  \kappa_2
 C_1^{1/2} (1 + s(t))^{1/2} t^{1/2} s(t)^{1/2} \\
& \leq & \frac{1}{4} \kappa_1 g_* y + C_3 (s(t) + s(t)^2) \quad
\mbox{ for } t \geq 0,
\end{eqnarray*}
where $C_3$ is some positive constant.

From the above inequalities we can get
\begin{eqnarray*}
  \frac{1}{2} \kappa_1 g_* t
 \leq  (M + \frac{1}{2} + C_3) |s(t)|^2 + (C_2 + C_3) s(t)
\quad \mbox{ for } t \geq 0.
\end{eqnarray*}
Now, let $t \geq 1$. In this case we see that
\begin{eqnarray*}
  \frac{1}{2} \kappa_1 g_* t
 \leq  (M + \frac{1}{2} + C_3) |s(t)|^2 + C_4 |s(t)|^2 +   \frac{1}{4} \kappa_1 g_* t
\quad \mbox{ for } t \geq 0,
\end{eqnarray*}
where $C_4$ is some positive constant. Thus it holds that
$$ \left(\frac{\kappa_1 g_*}{4(M + 1 + C_3 + C_4)} t\right)^{1/2} \leq s(t) \quad
          \mbox{ for } t \geq 1. $$
In case $0 \leq t \leq 1$, we have $s_0 \sqrt{t} \leq s(t)$.

Therefore, by putting $\nu_0 = \min\{ s_0, \left(\frac{\kappa_1
g_*}{4(M + 1 + C_3 + C_4)} \right)^{1/2}\}$ we conclude that
$$ \nu_0 \sqrt{t} \leq s(t) \quad \mbox{ for } t \geq 0. $$
\end{proof}

%%%%%%%%%%%%%%%%%%%%%%%%%%%%%%%%%%%%%%%%%%%%%%%%%%%%%
\section{Appearance of a moving carbonation front -- The case $s_0 =
0$}\label{front} The aim of  this section is to prove a result
concerning the existence of  weak solutions to P for the case  $s_0
= 0$. This is the case when the free boundary starts off moving
precisely from the  outer boundary [exposed to $CO_2$]. Before
giving the statement of the theorem, we denote for simplicity
$$ C_0((0, T]; X) = \{ z \in C([0,T]:X): z = 0 \mbox{ on } [0,\delta_z) \mbox{ for some } \delta_z >0\}. $$

\begin{theorem} \label{main2}
Let $T > 0$,  and $g$ and $h$ be functions on $[0,T]$ satisfying $g,
h \in W^{1,2}(0,T)$ and $g(t) \geq g_0 > 0$  and $h \geq 0$ for $t
\in [0,T]$, where $g_0$ is a given positive constant. Then under
(A1) there exists a triplet $\{s, u, v\}$ of functions such that $s
\in W^{1,\infty}(0,T)$, $s(0) = 0$, $s(t) > 0$ for $t \in (0,T]$,
$\bar{u}, \bar{v} \in L^{\infty}(Q(T))$, $\bar{u} - g, \bar{v} - h
\in L^2(0,T; X)$, $\bar{u}, \bar{v} \in C((0,T]; H)$, $\bar{u},
\bar{v} \in W^{1,2}_{loc}((0,T]; X^*)$,
\begin{equation}
s'(t) = \psi(\bar{u}(t,1)) \quad \mbox{ for a.e. } t \in [0,T],
\label{ss0}
\end{equation}
\begin{eqnarray}
 & &  \int_0^T \langle \bar{u}_t, z \rangle_X dt +
\int_{Q(T)}  \frac{\kappa_1}{s^2} \bar{u}_y z_y dydt
+ \int_0^T \frac{s'}{s} z(\cdot,1) dt \nonumber \\
& = & \int_{Q(T)} (f(\bar{u}, \bar{v})  - \frac{s'}{s} y\bar{u}_y) z
dy dt \quad
 \mbox{ for } z \in C_0((0, T]; X), \label{ss1}
\end{eqnarray}
\begin{eqnarray}
& &  \int_0^T \langle \bar{v}_t, z \rangle_X dt + \int_{Q(T)}
\frac{\kappa_2}{s^2} \bar{v}_y z_y
  dydt \nonumber
 \\
& = & - \int_{Q(T)} (f(\bar{u}, \bar{v})  + \frac{s'}{s} y \bar{v}_y
)z dy dt \quad
 \mbox{ for } z \in C_0((0, T]; X), \label{ss2}
\end{eqnarray}
where $\bar{u}$ and $\bar{v}$ are functions defined by
(\ref{change}).
\end{theorem}

\begin{proof}
First, let $\{s_{0n}\}$ be a sequence satisfying $s_{0n} > 0$ for
each $n$ and $s_{0n} \to 0$ as $n \to \infty$ and put $u_{0n} =
g(0)$  and $v_{0n} = h(0)$ on $[0,s_{0n}]$.  Then,  Theorem
\ref{previous} guarantees that P$(s_{0n}, u_{0n}, v_{0n}, g, h)$ has
a unique weak solution $\{s_n, u_n, v_n\}$ on $[0,T]$. Here, we
denote by $\bar{u}_n$ and $\bar{v}_n$ the functions defined by
(\ref{change}) with $s = s_n$, $u = u_n$ and $v = v_n$ for each $n$.
Since we can take positive constants $g^*$ and $h^*$ such that $g
\leq g^*$ and $h \leq h^*$ on $[0,T]$, $u_{0n} \leq g^*$ and $v_{0n}
\leq h^*$ on $[0,s_{0n}]$ for $n$ and $g^* = \gamma h^*$, Lemma
\ref{lem2} implies that
\begin{equation}
0 \leq u_n \leq g^*, 0 \leq v_n \leq h^* \quad \mbox{ on }
Q_{s_n}(T) \mbox{ for any } n. \label{bb}
\end{equation}
By (S3) and this shows that  $|s_n'(t)| \leq \psi(g^*)$ for $t \in
[0,T]$ and $n$ so that the set  $\{s_n\}$ is bounded in
$W^{1,\infty}(0,T)$. Clearly, there exists a positive constant $L_1$
such that $0 \leq s_n(t) \leq L_1$ for $t \in [0,T]$ and $n$.

Next, the following estimate is a direct consequence of \cite[Lemma
4.2]{Ai-Mun1}: For each $n$
\begin{eqnarray}
& &    \kappa_1 \int_0^t \int_0^{s_n(\tau)} |u_{nx}|^2 dx d\tau
  + \kappa_2 \int_0^t \int_0^{s_n(\tau)} |v_{nx}|^2 dx d\tau   \nonumber \\
& \leq & 2(C_f^2 + 1) \int_0^t \int_0^{s_n(\tau)} ( |u_n(\tau) -
g(\tau)|^2 +
         |v_n(\tau) - h(\tau)|^2) dx d\tau \nonumber  \\
& & + 2 \int_0^t s_n(\tau) ( |f(g(\tau),h(\tau))|^2 + |g_{\tau}(\tau)|^2 + |h_{\tau}(\tau)|^2) d\tau   \nonumber  \\
& & +  \int_0^t s_n'(\tau)(3|g(\tau)|^2 + |g(\tau)| + |h(\tau)|^2)
d\tau
      \quad     \mbox{ for } t \in  [0,T], \nonumber
\end{eqnarray}
where $C_f:= \beta \gamma$. Because of the boundedness of $\{s_n\}$
and (\ref{bb}) there exists a positive constant $M_2$ such that
$$ \int_0^T \int_0^{s_n(\tau)} |u_{nx}|^2 dxd\tau +
 \int_0^T \int_0^{s_n(\tau)} |v_{nx}|^2 dxd\tau \leq M_2 \mbox{ for } n. $$
Then, easily, we can obtain that $\{\bar{u}_{ny}\}$ and
$\{\bar{v}_{ny}\}$ are bounded in $L^2(Q(T))$.

From now on we provide the  estimate from below for the free
boundary  as follows. To do so from Lemma \ref{lem1} it follows that
\begin{eqnarray}
 \kappa_1  g_0 t & \leq & \int_0^{s_n(t)} x u_n(t) dx
+ \frac{1}{2} |s_n(t)|^2 + \kappa_1 \int_0^t u_n(\tau, s_n(\tau))
d\tau
   \nonumber \\
& & + \int_0^{s_n(t)} x v_n(t) dx
+ \kappa_2 \int_0^t \int_0^{s_n(\tau)} v_{nx}(\tau,x) dx d\tau \nonumber \\
& =: &  J_{1n} (t) + J_{2n}(t) + J_{3n}(t) + J_{4n}(t)  + J_{5n}(t)
\quad
 \mbox{ for } t \geq 0 \mbox{ and } n. \nonumber
\end{eqnarray}
Here, it is obvious that
$$ \left. \begin{array}{l}
J_{1n}(t) + J_{4n}(t) \leq \frac{1}{2} (g^* + h^*) |s_n(t)|^2 \\
 J_{5n}(t) \leq \kappa_2M_2^{1/2} (t s_n(t))^{1/2}
\end{array} \right\}
\mbox{ for } t \in [0,T]
  \mbox{ and } n.  $$
Similarly to (\ref{i2}), by using (\ref{FBC}) we observe that
\begin{eqnarray*}
J_{3n}(t)  \leq \frac{\kappa_1}{\alpha^{1/p}} \int_0^t
|s_n'(\tau)|^{1/p} d\tau \leq \frac{\kappa_1
T^{1-1/p}}{\alpha^{1/p}} s_n(t)^{1/p} \mbox{ for } t \in [0,T]
  \mbox{ and } n.
\end{eqnarray*}
From the above inequalities we have
\begin{eqnarray*}
& &  \kappa_1  g_0 t \\
& \leq &
 \frac{g^* + h^* + 1}{2} s_n(t)^2 +
\frac{\kappa_1 T^{1-1/p}}{\alpha^{1/p}} s_n(t)^{1/p}
+ \kappa_2(M_2T)^{1/2}  s_n(t)^{1/2} \\
& \leq &  \left(\frac{g^* + h^* + 1}{2} L_1^{2-\mu}  +
\frac{\kappa_1 T^{1-1/p}}{\alpha^{1/p}} L_1^{1/p - \mu}
+ \kappa_2(M_2T)^{1/2} L_1^{1/2 - \mu}\right)  s_n(t)^{\mu} \\
& =: & M_3 s_n(t)^{\mu} \mbox{ for } t \in [0,T]   \mbox{ and } n,
\end{eqnarray*}
where $$\mu := \min\{1/p, 1/2\}$$ so  that
\begin{equation}
 s_n(t) \geq \nu_1 t^{1/\mu} \mbox{ for } t \in [0,T]   \mbox{ and } n,
\label{ffc}
\end{equation}
where $\nu_1$ is a positive constant independent of $n$.

As next step, we wish to estimate the time derivative of
$\bar{u}_n$. Let $\delta > 0$ and $\eta \in L^2(\delta, T; X)$. Then
(S4) implies that
\begin{eqnarray*}
& &  |\int_{\delta}^{T} \langle \bar{u}_{nt}(t), \eta(t) \rangle_X dt| \\
& \leq &
 |\int_{\delta}^T \frac{\kappa_1}{s_n^2(t)} (\bar{u}_y(t),  \eta_y(t))_H dt|
+ |\int_{\delta}^{T}
 (\frac{s_n'(t)}{s_n(t)} \bar{u}_n(t,1) + \frac{s_n'(t)}{s_n(t)} ) \eta(t,1) dt|  \\
&   & + |\int_{\delta}^T (f(\bar{u}_n(t),\bar{v}_n(t)),\eta(t))_Hdt|
 + |\int_{\delta}^T \frac{s_n'(t)}{s_n(t)}( y \bar{u}_{ny}(t), \eta(t))_H dt| \\
& =: & I_{1n} + I_{2n} + I_{3n} + I_{4n}.
\end{eqnarray*}
Obviously, on account of (\ref{ffc}) it holds that
\begin{eqnarray*}
I_{1n} & \leq &
 \kappa_1 \int_{\delta}^{T} \frac{1}{\nu_1^2 t^{2/\mu}} |\bar{u}_{ny}(t)|_H |\eta_y(t)|_H dt \\
& \leq & \frac{\kappa_1}{\nu_1^2 \delta^{2/\mu}}
    |\bar{u}_{ny}|_{L^2(\delta, T; H)} |\eta_y(t)|_{L^2(\delta, T; H)};
\end{eqnarray*}
\begin{eqnarray*}
I_{2n} & \leq &
  \frac{\psi(g^*)}{\mu_1 \delta^{1/\mu}}(|\bar{u}_n|_{L^2(0,T; H^1(0,1))}+ T^{1/2})
 |\eta|_{L^2(\delta,T;X)};
\end{eqnarray*}
\begin{eqnarray*}
I_{3n} & \leq &
  \beta(\gamma h^* + g^*)  T^{1/2}  |\eta|_{L^2(\delta,T;X)};
\end{eqnarray*}
\begin{eqnarray*}
I_{4n} & \leq & \frac{\psi(g^*)}{\nu_1 \delta^{1/\mu}}
  |\bar{u}_{ny}|_{L^2(0,T;H)} |\eta|_{L^2(\delta,T;X)} \quad \mbox{ for } n.
\end{eqnarray*}
Hence, the set $\{\bar{u}_{nt}\}$ and $\{\bar{v}_{nt}\}$ are bounded
in $L^2(\delta,T; X^*)$
 for each $\delta >0$.

From these estimates we can take a subsequence $\{n_j\} \subset
\{n\}$ satisfying $s_{n_j} \to s$ weakly* in $W^{1,\infty}(0,T)$ and
$C([0,T])$, and $\bar{u}_{n_j} \to  \bar{u}$ and $\bar{v}_{n_j} \to
\bar{v}$ weakly* in $L^{\infty}(Q(T))$ and weakly in $L^2(0,T;
H^1(0,1))$, in $C([\delta,T]; H)$ and weakly in $ W^{1,2}(\delta, T;
X^*)$ for each $\delta > 0$   as $j \to \infty$, where $s \in
W^{1,\infty}(0,T)$, $\bar{u}, \bar{v} \in L^{\infty}(Q(T))$,
$\bar{u} - g, \bar{v} - h \in L^2(0,T; X)$ and $$\bar{u}, \bar{v}
\in C((0,T]; H) \cap W_{loc}^{1,2}((0, T]; X^*).$$
 By (\ref{ffc}) we
have $s(t) > 0 $ for $t > 0$. Also, $s(0) = 0$.

In order to complete the proof of the Theorem, it is necessary to
show that (\ref{ss1}), (\ref{ss2}) and (\ref{ss0}) hold. Let $\eta
\in C_0((0,T]; X)$. Then $\eta = 0$ on $[0,\delta]$ for some $\delta
> 0$. By taking $z = \eta$ in (S4) we infer that

\begin{eqnarray*}
& & \int_{\delta}^{T} \langle \bar{u}_{n_j t}(t), \eta(t) \rangle_X
dt
+ \int_{\delta}^T \frac{\kappa_1}{s_{n_j}^2(t)} (\bar{u}_{n_j y}(t), \eta_y(t) )_H dt \\
& & + \int_{\delta}^{T}
 (\frac{s_{n_j}'(t)}{s_{n_j}(t)} \bar{u}_{n_j}(t,1)
     + \frac{s_{n_j}'(t)}{s_{n_j}(t)} ) \eta(t,1) dt  \\
&  = & \int_{\delta}^T (f(\bar{u}_{n_j}(t),\bar{v}_{n_j}(t))
       + \frac{s_{n_j}'(t)}{s_{n_j}(t)} y \bar{u}_{n_j y}(t),  \eta(t))_H dt \quad
\mbox{ for } j.
\end{eqnarray*}
Elementary calculations yield:
\begin{equation}
\bar{u}_{n_j}(\cdot,1) \to \bar{u}(\cdot,1) \mbox{ in }
L^4(\delta,T) \mbox{  as } j \to \infty \label{con}
\end{equation}
so that
$$ \int_{\delta}^{T}
 \frac{s_{n_j}'(t)}{s_{n_j}(t)} \bar{u}_{n_j}(t,1)   \eta(t,1) dt  \to
\int_{\delta}^{T}
 \frac{s'(t)}{s(t)} \bar{u}(t,1)   \eta(t,1) dt \quad \mbox{ as } j \to \infty.
$$
Moreover, we can obtain $s_{n_j}' \to s'$ in $L^4(\delta, T)$ as $j
\to \infty$ and
$$
\int_{\delta}^T (\frac{s_{n_j}'(t)}{s_{n_j}(t)} y \bar{u}_{n_j
y}(t),  \eta(t))_H dt \to \int_{\delta}^T (\frac{s'(t)}{s(t)} y
\bar{u}_{y}(t),  \eta(t))_H dt \mbox{ as } j \to \infty. $$
Therefore, we can prove that (\ref{ss1}) holds. Similarly,
(\ref{ss2}) is valid. Finally, by (\ref{con}) we get (\ref{ss0}).
Thus the proof of this theorem has been finished.
\end{proof}
\section*{A practical comment} Before using  in the engineering
practice the $\sqrt{t}$ information for forecasting purposes, the
practitioner should be aware of the fact that its validity is
closely related to the validity of the underlying free-boundary
model P. Relying on our working experience with such FBPs for
carbonation (based on \cite{Chem}, e.g.), we can say that P captures
well accelerated carbonation tests, but it may not be suitable for
predicting the evolution of carbonation scenarios under natural
exposure conditions.

\section*{Acknowledgment} The initial phase of this research has been supported
 by the German Science Foundation (DFG), special program SP 1122.

\bibliographystyle{plain}

\end{document}